\documentclass[12pt,draftcls,journal,onecolumn]{IEEEtran}
\usepackage{amsfonts,color,morefloats,pslatex,graphicx,graphics}
\usepackage{amssymb,amsthm, amsmath,latexsym}

\newtheorem{theorem}{Theorem}
\newtheorem{lemma}[theorem]{Lemma}
\newtheorem{remark}[theorem]{Remark}
\newtheorem{proposition}[theorem]{Proposition}

\newtheorem{example}[theorem]{Example}

\newcommand{\C}{{\mathcal{C}}}

\newcommand{\bF}{{\mathbb F}}

\usepackage{blindtext}

\ifCLASSINFOpdf

\else

\fi

\begin{document}
%
\title{ The weight hierarchies of three classes of linear codes
\thanks{
}
}
\author{Wei Lu,\,\, Qingyao Wang,\,\, Xiaoqiang Wang,\,\, Dabin Zheng{\thanks{W. Lu, Q. Wang, X. Wang, D. Zheng are with Hubei Key Laboratory of Applied Mathematics, Faculty of Mathematics and Statistics, Hubei University, Wuhan 430062, China. Email:  weilu23@hubu.edu.cn; 2300805020@qq.com; waxiqq@163.com; dzheng@hubu.edu.cn.
 The corresponding author is Dabin Zheng. The work of D. Zheng was supported in part by the National Natural Science Foundation of China under the Grant Numbers 62272148, 11971156, and the work of X. Wang was supported in part by   Natural Science Foundation of Hubei Province of China under the Grant Number 2023AFB847, and Sunrise Program of Wuhan under the Grant Number 2023010201020419, and
 The work of W. Lu was supported by Postdoctoral Fellowship Program of CPSF under grant No.GZC20230750.
}}
}

%
%
%
%

\maketitle

\begin{abstract}
Studying the generalized Hamming weights of linear codes is a significant research area within coding theory, as it provides valuable structural information about the codes and plays a crucial role in determining their performance in various applications. However, determining the generalized Hamming weights of linear codes, particularly their weight hierarchy, is generally a challenging task. In this paper, we focus on investigating the generalized Hamming weights of three classes of linear codes over finite fields. These codes are constructed by different defining sets. By analysing the intersections between the definition sets and the duals of all $r$-dimensional subspaces, we get the inequalities on the sizes of these intersections. Then constructing subspaces that reach the upper bounds of these inequalities, we successfully determine the complete weight hierarchies of these codes.
\end{abstract}



%
\IEEEpeerreviewmaketitle

\section{Introduction}\label{sec-auxiliary}

The concept of generalized Hamming Weights (GHWs) was first introduced by Helleseth et al.~\cite{Hellesethetal1977} and Kl{\o}ve~\cite{Klove1978} as a natural extension of the minimum Hamming distance of linear codes. GHWs provide fundamental information about linear codes and have significant implications in various applications. In 1991, Wei~\cite{Wei1991} presented a series of remarkable results on GHWs, using them to characterize the cryptography performance of linear codes in wire-tap channels of type II. Since then, GHWs of linear codes have become an interesting research topic in both theory and applications. They have been utilized in dealing with t-resilient functions, trellis or branch complexity of linear codes~\cite{Chor1985,Geometric}, computation of state and branch complexity profiles~\cite{Forney1994,Kasamietal1993}, shortening linear codes efficiently~\cite{Hellesethetal19951}, determination of erasure list-decodability~\cite{Guruswami2003}, and more.

Significant progress has been made over the past two decades in the study of GHWs of linear codes. Various lower and upper bounds for GHWs of some linear codes have been established by~\cite{Ashikhmin1999,Cohen1994,Hellesethetal19951}. GHW values have been determined or estimated for many classes of linear codes, including Hamming codes~\cite{Wei1991}, Reed-Muller codes~\cite{Heijnen1998,Wei1991}, binary Kasami codes~\cite{Hellesethetal19952}, Melas and dual Melas codes~\cite{vandergeer1994}, certain BCH codes and their duals \cite{Cheng1997,Duursma1996,Feng1992,Moreno1998,Shim1995,vandergeer19942,vandergeer1995}, certain trace codes~\cite{Cherdien2001,Stichtenoth1994,vandergeer19952,vandergeer1996}, certain algebraic geometry codes \cite{Yang1994}, and some cyclic codes \cite{Lishuxing2017,Ge2016,Yang2015}. However, determining the complete weight hierarchies of linear codes remains a challenging task, and to the best of our knowledge, the complete weight hierarchies of only a few linear codes are known.

For a set $D=\{d_1,d_2,\ldots, d_n\}\subset {\mathbb F}_{q^m}$, Ding and Niederreiter~\cite{Ding2007} proposed a generic method for constructing linear codes as follows:
\begin{equation}\label{linearcode1}
\C_D=\left\{ {\bf c}(x) = \left( {\rm Tr}_1^m( xd_1), {\rm Tr}_1^m(xd_2), \ldots, {\rm Tr}_1^m(xd_n) \right): x\in {\mathbb F}_{q^m} \right\},
\end{equation}
where ${\rm Tr}_1^m(\cdot)$ is the trace function from $\bF_q^m$ to $\bF_q$.
Then $\C_D$ is a linear code over $\bF_q$ with length $n$. The set $D$ is called the defining set of $\C_D$. This construction is fundamental since every liner code over $\bF_q$ can be represented as $\C_D$ for some defining set $D$~\cite{Xiang2016}.
In recent years, Ding and Niederreiter's construction was extended to the bivariate form, namely,
\begin{equation}\label{linearcode2}
\C_D = \left\{  ({\rm Tr}_1^{m}(ax)+{\rm Tr}_1^{k}(by))_{(x,y) \in D}:a \in {\mathbb F}_{q^m},b\in {\mathbb F}_{q^k} \right\},
\end{equation}
where $D\subset {\mathbb F}_{q^m}\times {\mathbb F}_{q^k}$. There has been extensive research on the parameters and Hamming weight distributions of linear codes constructed using definition sets, leading to fruitful results. Li and Mesnager~\cite{LiMesnager2020} provided a survey on linear codes constructed from cryptographic functions.

Recently, attention has been paid to the GHWs of linear codes constructed from defining sets. Jian et al.~\cite{Jian2017} were the first to investigate GHWs of linear codes constructed from skew sets. Subsequently, Li~\cite{Li2018}, Li~\cite{Li2021}, Liu et al.~\cite{LiuZheng2023}, Liu and Wang~\cite{LiuWang2020}, Li and Li~\cite{LiLi2021,LiLi2022}, and Li et al.~\cite{Likangquan2022} explored GHWs of specific classes of linear codes constructed from defining sets associated with quadratic functions, cyclotomic classes, or cryptographic functions. While weight distributions of linear codes constructed using definition sets have been extensively studied, the knowledge about the weight hierarchies (the definition given in Section~\ref{sec-auxiliary}) of this class of linear codes is still limited. Even if the weight distribution of linear codes is known, it is both interesting and challenging to find a suitable definition set of linear codes and thus determine the Hamming weight hierarchy of the corresponding linear codes. In this paper, we determine the weight hierarchies of three classes of linear codes, and the weight distributions of these codes have been given in \cite{Hu2021,Hu2022}. Specially, the constructions of the codes are as follows:
\begin{enumerate}
\item[$\bullet$]
The first class of linear codes $\mathcal{C}_D$ is defined in (\ref{linearcode1}) with defining set $D= \bF_{q^m}\setminus \Omega$ and $\Omega  = \bigcup_{i=0}^{h} (\theta_i+\bF_{q^{k}})$, where $1\le k< m $, $k\,|\,m$, $0\le h\le q-1$, $\theta_0 = 0$,  $\theta_i \in \bF_{q^m}^*$ and $\theta_i - \theta_j \notin \bF_{q^k}$ for $0\le j < i \le h$.
\item[$\bullet$]  The second class of linear codes $\mathcal{C}_D$ is defined in (\ref{linearcode2}) with defining set
$$
D=\{(x,y)\in {\mathbb F}_{q^m}\times {\mathbb F}_{q^k}: x \in {\mathbb F}_{q^m}\setminus {\mathbb F}_{q^s}, y \in {\mathbb F}_{q^k}\setminus {\mathbb F}_{q^l}\},
$$
where $0<s<m$, $0<\ell<k$, $s|m$, $\ell \, |\, k$ and $k-\ell\le m-s$.
\item[$\bullet$] The third class of linear codes $\mathcal{C}_D$ is defined in (\ref{linearcode2}) for $q=2$ and $k=m$ with defining set
$$
D=\{(x,y) \in \bF_{2^m}^2 : ({\rm Tr}_1^m (x(x+y)), {\rm Tr}_1^m(y(x+y))) = (0,1) \}.
$$
\end{enumerate}

The key to calculating the $r$-th generalized Hamming weight of these linear codes is to determine the maximum value of the intersection between the definition sets and the duals of all $r$-dimensional subspaces for a positive integer $r$. The sizes of these intersections can be expressed as some exponential sums over finite fields. By analyzing these exponential sums in detail we get the inequalities for exponential sums. Then, with these inequalities, we construct subspaces within the discussed codes that achieve this established boundary, then determine the $r$-th generalized Hamming weight of discussed linear codes.

%

The rest of this paper is organized as follows: In Section 2, we introduce the linear codes from the defining-sets and their generalized Hamming weights. Section 3 determines the weight hierarchy of three classes of linear codes from defining sets. These defining sets are of the unitary form or the bivariate form, respectively. Finally, Section 4 concludes the paper.

\section{Preliminaries}\label{sec-auxiliary}

Let $q$ be a prime number and ${\mathbb F}_q$ be a finite field with $q$ elements. An $[n,k,d]$ linear code $\C$ over the finite field  $\bF_q$  is a $k$-dimensional subspace of $\bF_q^n$
with length $n$ and minimum Hamming distance $d$. For each $r$ with $1\leq r\leq k$, let $[\C, r]$ denote the set of all $r$-dimensional $\bF_q$-subspaces of $\C$. For each $H\in [\C, r]$, the {\em support} of $H$, denoted by supp($H$), is the set of not-always-zero bit positions of vectors in $H$, i.e.,
$$ {\rm supp}(H)=\left\{ i \, :\, 1\leq i\leq n, \,\, c_i\neq 0  \,\, {\rm for \,\, some} \,\, (c_1,c_2,\ldots,c_n)\in H \right\}.$$
The $r$-th {\em generalized Hamming weight} (GHW) of $\C$ is defined by
	$$ d_r(\C)= {\rm min} \left\{ \,|{\rm supp}(H)|\,:\, H \in [ \C, r]\, \right\}.$$  The set $\left\{ d_1(\C), d_2(\C),\cdots, d_k(\C)\right\}$ is called the {\em weight hierarchy}
of $\C$. It is worth noting that $d_1(\C)$ is just the minimum Hamming distance of~$\C$.

\subsection{Linear codes with the unitary form}

Let $\C_{D}$ be a linear code defined in (\ref{linearcode1}) with dimension $k$. It is easy to see that $k=m- {\rm dim}_{\bF_q} K$, where
\begin{equation}\label{eq:0522}
K=\left\{ x \in {\mathbb F_{q^m}}\, :\, ( {\rm Tr}_1^m(dx) )_{d \in D} = {\bf 0} \right\}.
\end{equation}
Let $H$ be a subspace of $\bF_{q^m}$, and its dual is defined as follows:
\begin{equation}\label{eq:052201}
H^\perp=\{x\in \bF_{q^m}: {\rm Tr}_1^m(xh)=0\  {\rm for\ all}\ h\in H\}.
\end{equation}
It is easy to show that
${\rm dim}_{\mathbb{F}_q}(H)+{\rm dim}_{\mathbb{F}_q}(H^{\perp})=m$ and $(H^{\perp})^{\perp}=H$.
 In \cite{Likangquan2022}, Li, Chen and Qu have introduced a method to compute the generalized Hamming weight of $\C_
D$, which is outlined as follows.
\begin{proposition}\cite{Likangquan2022}\label{GHW}
Let $K$ be as mentioned in (\ref{eq:0522}), and let $\C_{D}$ be a linear code defined in (\ref{linearcode1}) with dimension $m-{\rm dim}_{\bF_q} K$. Then, for any integer $1\leq r\leq m-{\rm dim}_{\bF_q} K$, the $r$-th generalized Hamming weight of $\C_{D}$ is
\begin{equation}\label{eq:ghw}
d_r(\C_D) =n-{\rm max}\, \left\{ |D\cap H^\perp| : H \in [{\mathbb F_{q^m}}, r]\,\, {\rm and} \,\, H\cap K=\{0 \} \right\},
\end{equation}
where $[{\mathbb F}_{q^m}, r]$ denotes the set of all $r$-dimensional subspaces of ${\mathbb F}_{q^m}$.
\end{proposition}

\begin{lemma}(Corollary 1.1.25 in \cite{Algebraic_Geometric})
\label{lm:charactersum1}
Let $H$ be a subspace of $\bF_{q^m}$. Let $\theta\in\bF_{q^m} $ and $\zeta_q$ be a q-th primitive root of unity, then we have
$$
\sum_{x\in H}\zeta_q^{{\rm Tr}_1^m(\theta x)}=
\begin{cases}
 0,& \text{$\theta\notin H^\perp$},
\\
|H|,& \text{$\theta\in H^\perp$}.
\end{cases}
$$
\end{lemma}

\subsection{Linear codes with  the bivariate form}
Let $H$ be a subspace of ${\mathbb F}_{q^m}\times{\mathbb F}_{q^k}$. Similar to (\ref{eq:052201}), we define the dual of $H$ as follows,
\begin{equation}\label{eq:xy0511}
H^{\perp}=\left\{ (x,y)\in {\mathbb F}_{q^m}\times{\mathbb F}_{q^k}:{\rm Tr}_1^m(\alpha x)+{\rm Tr}_1^k(\beta y)=0\ {\rm for\ any}\ (\alpha,\beta)\in H \right\}.
\end{equation}
It is clear that $H^{\perp}$ is a subspace of ${\mathbb F}_{q^m}\times{\mathbb F}_{q^k}$.  Assume that $H$ is an $r$-dimensional $\bF_q$-subspace   of ${\mathbb F}_{q^m}\times{\mathbb F}_{q^k}$ and $\{(\alpha_1,\beta_1), (\alpha_2,\beta_2),\dots, (\alpha_r,\beta_r)\}$ is an ${\mathbb F}_q$-basis of $H$. Then
\begin{eqnarray*}
|H^{\perp}|&=&\sum_{x\in\bF_{q^m}}\sum_{y\in\bF_{q^k}}\prod_{i=1}^r\left(\sum_{v_i\in\bF_q}\frac{1}{q}\zeta_q^{v_i({\rm Tr}_1^m(\alpha_i x)+{\rm Tr}_1^k(\beta_i y))}\right)\\
                     &=&\frac{1}{q^r}\sum_{\left ( \alpha ,\beta  \right ) \in H}\sum_{x\in\bF_{q^m}}\sum_{y\in\bF_{q^k}}\zeta_q^{{\rm Tr}_1^m(\alpha x)+{\rm Tr}_1^k(\beta y)}\\
                     &=&q^{m+k-r}.
\end{eqnarray*}
Hence, we have ${\rm dim}_{\mathbb{F}_q}(H^{\perp})=m+k-r$. It should be noticed that ${\rm dim}_{\mathbb{F}_q}(H)+{\rm dim}_{\mathbb{F}_q}(H^{\perp})=m+k$ and $(H^{\perp})^{\perp}=H$. With an analysis similar to Lemma~\ref{lm:charactersum1}, we have the following result.
\begin{lemma}
\label{lm:charactersum101}
Let $H$ be a subspace of ${\mathbb F}_{q^m}\times\bF_{q^k}$ and $\zeta_q$ be a q-th primitive root of unity, then we have
\begin{equation*}
\sum_{( \alpha ,\beta ) \in H}\zeta_q^{{\rm Tr}_1^m(\alpha x)+{\rm Tr}_1^k(\beta y)}=
\begin{cases}
 0,& \text{$(x,y)\notin H^\perp$},
\\
|H|,& \text{$(x,y)\in H^\perp$}.
\end{cases}
\end{equation*}
\end{lemma}

Let
\begin{equation}\label{eq:xy0522}
K=\left\{ (x,y) \in {\mathbb F}_{q^k}\times{\mathbb F}_{q^m} \, :\, ( {\rm Tr}_1^m(d_1x)+{\rm Tr}_1^k(d_2y) )_{(d_1,d_2) \in D} = {\bf 0} \right\}.
\end{equation}
Then, the $r$-th generalized Hamming weight of linear codes given in (\ref{linearcode2}) is determined in the following proposition. This is a generalization of Proposition~\ref{GHW}.
\begin{proposition}{\label{GHW2}}
Let $\C_{D}$ be a linear code defined in (\ref{linearcode2}) with dimension $m-{\rm dim}_{\bF_q} K$, where $K$ is given in (\ref{eq:xy0522}) and  $D\subset {\mathbb F}_{q^m}\times {\mathbb F}_{q^k}$. Then, for any integer $1\leq r\leq m-{\rm dim}_{\bF_q} K$, the $r$-th generalized Hamming weight of $\C_{D}$ is
\begin{equation}\label{eq:rghw}
d_r(\C_D) =n-{\rm max}\, \left\{ |D\cap H^\perp| : H \in [{\mathbb F}_{q^m}\times\bF_{q^k}, r]\,\, {\rm and} \,\, H\cap K=\{(0,0) \} \right\},
\end{equation}
where $[{\mathbb F}_{q^m}\times\bF_{q^k}, r]$ denotes the set of all $r$-dimensional subspaces of ${\mathbb F}_{q^m}\times\bF_{q^k}$.
\end{proposition}

\begin{proof} By the definition of $\C_{D}$, it is clear that  the dimension of $\C_{D}$ is $m-{\rm dim}_{\bF_q} K$, we omit the details of the proof. In the following, we show the  $r$-th generalized Hamming weight of this code.

On one hand, let $L$ be any $r$-dimensional $\bF_q$-subspace of $\C_D$ and $\{v_1,\dots, v_r\}$ be an $\bF_q$-basis of $L$. By the definition of $\C_{D}$, we know that $v_i$ can be expressed as
 $$v_i=( {\rm Tr}_1^m(d_1\alpha_i)+{\rm Tr}_1^k(d_2\beta_i) )_{(d_1,d_2) \in D}$$
for $1\le i\le r $, where $(\alpha_i,\beta_i)\in \bF_{q^m}\times\bF_{q^k}$. Let $H_r$ be the linear span of the vectors $\{(\alpha_1,\beta_1),\dots,(\alpha_r,\beta_r)\}$. It is clear that $H_r$ is an $r$-dimensional $\bF_q$-subspace of ${\mathbb F}_{q^m}\times\bF_{q^k}$ and $ H_r\cap K=\{(0,0) \}$. By the definition of $L$ and $H_r$, we have
  $$L=\{( {\rm Tr}_1^m(d_1x)+{\rm Tr}_1^k(d_2y) )_{(d_1,d_2) \in D}\in\C_D:(x,y)\in H_r\}.$$
Since $L$ is any $r$-dimensional $\bF_q$-subspace of $\C_D$, we have
\begin{equation}\label{eq:CH}
\left\{ \C(H) : H \in [{\mathbb F}_{q^m}\times\bF_{q^k}, r]\,\, {\rm and} \,\, H\cap K=\{(0,0) \} \right\}\supset [\C_D,r],
\end{equation}
where
$$\C(H)=\{( {\rm Tr}_1^m(d_1x)+{\rm Tr}_1^k(d_2y) )_{(d_1,d_2) \in D}\in\C_D:(x,y)\in H\}.$$

On the other hand, we choose an $r$-dimensional $\bF_q$-subspace $H'_r$ of ${\mathbb F}_{q^m}\times\bF_{q^k}$ such that $H'_r\cap K=\{(0,0)\}$. Since $K$ is the kernel of the linear homomorphism
\begin{equation*}
\begin{split}
\phi:\,\,\,  &{\mathbb F}_{q^m}\times\bF_{q^k}  \longrightarrow  \C_D \\
   &     (x,y) \,\,\,\,\,\,\,\,\,\,\,\,\,\,\,\,\longmapsto ( {\rm Tr}_1^m(d_1x)+{\rm Tr}_1^k(d_2y) )_{(d_1,d_2) \in D}\mbox{,}
\end{split}
\end{equation*}
we have that $\phi|_{H_r'}$ is an injection. This implies that $\C(H'_r)=\phi(H'_r)$ is an $r$-dimensional $\bF_q$-subspace of $\C_D$. Thus, from (\ref{eq:CH}) we have
$$
\left\{ \C(H) : H \in [{\mathbb F}_{q^m}\times\bF_{q^k}, r]\,\, {\rm and} \,\, H\cap K=\{(0,0) \} \right\}= [\C_D,r].
$$
By the definition of generalized Hamming weights, we have
\begin{equation}\label{eqn:commonzero7}
    d_r(\C_D) ={\rm min}\, \left\{ |{\rm supp}(\C(H))| : H \in [{\mathbb F_{q^m}},r]\,\, {\rm and} \,\, H\cap K=\{(0,0) \} \right\}.
\end{equation}
For any $r$-dimensional $\bF_q$-subspace $H$ of ${\mathbb F}_{q^m}\times\bF_{q^k}$, we let $\{(\alpha_1,\beta_1), (\alpha_2,\beta_2),\dots, (\alpha_r,\beta_r)\}$ be an ${\mathbb F}_q$-basis of $H$. Then
\begin{equation}\label{eq:0527}
\begin{split}
|{\rm supp}(\C(H))|=&n-|\{1\le i\le n: c_i=0 {\rm\ for\ any\ } (c_1,\dots,c_n)\in\C(H)\}|\\
                     =&n-\sum_{(x,y)\in D}\prod_{i=1}^r\left(\sum_{v_i\in\bF_q}\frac{1}{q}\zeta_q^{v_i({\rm Tr}_1^m(\alpha_i x)+{\rm Tr}_1^k(\beta_i y))}\right)\\
                     =&n-\frac{1}{q^r}\sum_{(x,y)\in D}\sum_{\left ( \alpha ,\beta  \right ) \in H}\zeta_q^{{\rm Tr}_1^m(\alpha x)+{\rm Tr}_1^k(\beta y)}.
\end{split}
\end{equation}
For any $(x,y)\in \bF_{q^m}\times\bF_{q^k}$, we have
\begin{equation}\label{eqn:commonzero6}
\sum_{( \alpha ,\beta ) \in H}\zeta_q^{{\rm Tr}_1^m(\alpha x)+{\rm Tr}_1^k(\beta y)}=
\begin{cases}
 0,& \text{$(x,y)\notin H^\perp$},
\\
|H|,& \text{$(x,y)\in H^\perp$}.
\end{cases}
\end{equation}
Combining (\ref{eq:0527}) and (\ref{eqn:commonzero6}), we have
$$
|{\rm supp}(\C(H))|=n-|D\cap H^\perp|.
$$
From (\ref{eqn:commonzero7}), the $r$-th generalized Hamming weight of $\C_D$ is
$$
d_r(\C_D) =n-{\rm max}\, \left\{ |D\cap H^\perp| : H \in [{\mathbb F}_{q^m}\times\bF_{q^k}, r]\,\, {\rm and} \,\, H\cap K=\{(0,0) \} \right\}.
$$
 This completes the proof.
\end{proof}

\section{ The weight hierarchies of three classes of linear Codes}
In this section, we will give the weight hierarchy of three classes of linear codes. Our  method of proof is as follows.
By analysing the defining sets of linear codes, we can establish some bounds for their generalized Hamming weights. Subsequently, by constructing subspaces of the discussed codes that match these bounds, we determine their weight hierarchies. As will be seen later, the results for these cases are distinct. Hence, the proof should be treated separately.

\subsection{Linear codes from  the defining set $\bF_{q^m}\setminus\bigcup_{i=0}^{h} (\theta_i+\bF_{q^{k}})$}
In this subsection, we investigate a class of linear codes denoted as follows:
\begin{equation}
\label{eqn:linearcode1}
\C_D = \{\mathbf{c}(a)=({\rm Tr}_1^{m}(ax))_{x \in D}:a \in \bF_{q^m} \},
\end{equation}
where the defining set is given by
$D= \bF_{q^m}\setminus \Omega$ and $\Omega  = \bigcup_{i=0}^{h} (\theta_i+\bF_{q^{k}})$.
Here, we always assume that $1\le k< m $, $k\,|\,m$, $0\le h\le q-1$, $\theta_0 = 0$,  $\theta_i \in \bF_{q^m}^*$ and $\theta_i - \theta_j \notin \bF_{q^k}$ for $0\le j < i \le h$. Recall that
\begin{equation*}
H^\perp=\{x\in \bF_{q^m}: {\rm Tr}_1^m(xh)=0\  {\rm for\ all}\ h\in H\}
\end{equation*}
for any $\bF_q$-subspace $H$ of $\bF_{q^m}$. If $H=\bF_{q^k}$, it is clear that ${\rm Tr}_1^m(xh)={\rm Tr}_1^k(h{\rm Tr}_k^m(x))$ for any $h \in H$.
Then we have
$$
\bF_{q^k}^\perp=\{x\in \bF_{q^m}: {\rm Tr}_k^m(x)=0\}.
$$
We now show the following results, which is useful for us to get the weight hierarchy of $\C_D$.
\begin{lemma}
\label{lm:unemptyset}
Let the symbols be given as above. Then there exists an $x\in \bF_{q^k}^\perp$ such that ${\rm Tr}_1^m(\theta_i x)\ne 0$ for any $1\le i \le h$.
\end{lemma}
\begin{proof}
Suppose that for any $x\in \bF_{q^k}^\perp$, there exists an integer~$i$ with $1\leq i\leq h$ such that ${\rm Tr}_1^m(\theta_i x)= 0$. Then we have $\bigcup_{i=1}^h\{x\in\bF_{q^k}^\perp:{\rm Tr}_1^m(\theta_i x)= 0\}\supset\bF_{q^k}^\perp$, which implies that
\begin{equation}\label{eq:050801}
\sum_{i=1}^h|\{x\in\bF_{q^k}^\perp:{\rm Tr}_1^m(\theta_i x)= 0\}|\ge |\bF_{q^k}^\perp|.
\end{equation}
By the definition of $\bF_{q^k}^\perp$, we obtain
\begin{equation}\label{eq:0508}
\begin{split}
\sum_{i=1}^h|\{x\in\bF_{q^k}^\perp\,:\,{\rm Tr}_1^m(\theta_i x)= 0\}|&=\sum_{i=1}^h|\{x\in\bF_{q^m}\,:\,{\rm Tr}_1^m(\theta_i x)= 0\,\, and\,\,{\rm Tr}_k^m(x)=0 \}|\\
&=\frac{1}{q^{k+1}}\sum_{i=1}^{h}\sum_{x\in\bF_{q^m}}\sum_{y \in \mathbb{F}_q}\zeta_q^{y {\rm Tr}_1^m(\theta_i x)}\sum_{z \in \mathbb{F}_{q^k}}\zeta_q^{{\rm Tr}_1^m(z x)}\\
&=\frac{1}{q^{k+1}}\sum_{i=1}^{h}\sum_{x\in\bF_{q^m}}\sum_{z \in \mathbb{F}_{q^k}}\zeta_q^{{\rm Tr}_1^m((y\theta_i+z) x)}.
\end{split}
\end{equation}
Since $\theta_i \notin \mathbb{F}_{q^k}$ for  $1\leq i\leq h$, we have $y\theta_i+z=0$ if and only if $y=z=0$. Then from (\ref{eq:0508}), it is clear that
\begin{equation}
\begin{split}
\sum_{i=1}^h|\{x\in\bF_{q^k}^\perp\,:\,{\rm Tr}_1^m(\theta_i x)= 0\}|=h\cdot q^{m-k-1}\ge q^{m-k}=|\bF_{q^k}^\perp|
\end{split}
\end{equation}
since $h\le q-1$, which is contradictory to (\ref{eq:050801}).
 This completes the proof.
\end{proof}

In the following, we give the weight hierarchy of the first class of linear codes.

\begin{theorem}\label{th:GHW1}
Let $\C_D$ be a linear code defined in (\ref{eqn:linearcode1}). Then $\C_D$ is a $[q^m-(h+1)q^k,m,q^m-(h+1)q^k-q^{m-1}+(h+1)q^{k-1}]$ code and has the following weight hierarchy:
$$ d_r(\C_D) =\begin{cases}
 q^m-(h+1)q^k - q^{m-r}+(h+1)q^{k-r},& \text{$1\le r\le k $},
\\
q^m-(h+1)q^k-q^{m-r}+1,& \text{$k<r \le m$}.
\end{cases}$$
\end{theorem}

\begin{proof}
From \cite{Hu2022}, we know that the code $\C_{D}$ has parameters $[q^m-(h+1)q^k,m,q^m-(h+1)q^k-q^{m-1}+(h+1)q^{k-1}]$. Then
 the subspace
$$
K=\left\{ x \in {\mathbb F_{q^m}}\, :\, ( {\rm Tr}_1^m(dx) )_{d \in D} = {\bf 0} \right\}
$$
is $\{0\}$. By Proposition \ref{GHW}, the $r$-th generalized Hamming weight of $\C_{D}$ is
\begin{equation*}
d_r(\C_D) =n-{\rm max}\, \left\{ |D\cap H^{\perp}| : H \in [{\mathbb F_{q^m}}, r] \right\},
\end{equation*}
where $1\leq r\leq m$. In the next, we determine the maximum value of $|D\cap H^{\perp}|$ for $H \in [{\mathbb F_{q^m}}, r]$.

Let $H_r$ be an $r$-dimensional subspace of $\bF_{q^m}$. Then from Lemma \ref{lm:charactersum1}, we know that $|D\cap H_r^{\perp}|$ can be expressed as follows:
\begin{equation}\label{eq:050802}
\begin{split}
|D\cap H_r^{\perp}|=&\frac{1}{q^r}\sum_{x \in D}\sum_{\omega \in H_r}\zeta_q^{{\rm Tr}_1^{m}(\omega x)}\\
=&\frac{1}{q^r}\sum_{x \in \bF_{q^m}}\sum_{\omega \in H_r}\zeta_q^{{\rm Tr}_1^{m}(\omega x)}-\frac{1}{q^r}\sum_{x \in \bigcup_{i=0}^h(\theta_i+ \bF_{q^k})}\sum_{\omega \in H_r}\zeta_q^{{\rm Tr}_1^{m}(\omega x)}\\
=& q^{m-r} - \frac{1}{q^r} \sum_{i=0}^h \sum_{\omega \in H_r}\zeta_q^{{\rm Tr}_1^{m}(\omega \theta_i)} \sum_{x \in \bF_{q^k}}\zeta_q^{{\rm Tr}_1^{m}(\omega x)}.
\end{split}
\end{equation}
It is clear that
$$
\sum_{x \in \bF_{q^k}}\zeta_q^{{\rm Tr}_1^{m}(\omega x)}=\sum_{x\in\bF_{q^{k}}}\zeta_q^{{\rm Tr}_1^k(x{\rm Tr}_k^m(\omega ))}=\begin{cases}
q^k,& \text{$\omega \in \bF_{q^k}^ \perp $},
\\
0,& \text{$\omega \notin \bF_{q^k}^ \perp $}.
\end{cases}
$$
Then from (\ref{eq:050802}) we have
\begin{equation}\label{eq:050803}
|D\cap H_r^{\perp}|= q^{m-r} - q^{k-r} \sum_{i=0}^h \sum_{\omega \in H_r \cap \bF_{q^k}^ \perp}\zeta_q^{{\rm Tr}_1^{m}(\omega \theta_i)}=q^{m-r} - q^{k-r} \sum_{i=0}^hN_i,
\end{equation}
where $N_i=\sum_{\omega \in H_r \cap \bF_{q^k}^ \perp}\zeta_q^{{\rm Tr}_1^{m}(\omega \theta_i)}$.
From Lemma \ref{lm:charactersum1} we obtain
\begin{equation}\label{eq:050804}
N_i=
\begin{cases}
0,& \text{$\theta_i\notin(H_r\cap \bF_{q^k}^ \perp)^\perp$},\\
|H_r\cap \bF_{q^k}^ \perp|,& \text{$\theta_i\in(H_r\cap \bF_{q^k}^ \perp)^\perp$}
\end{cases}
\end{equation}
for $0\le i\le h$. We now show the maximum value of $|D\cap H_r^{\perp}|$ for $1\le r\le m$ from the following two cases.

\noindent {\bf Case 1:} $1\le r\le k$. Since ${\rm dim}(\bF_{q^k}^ \perp)=m-k$ and $m-k+r\le m$, there exists an $r$-dimensional subspace $H_r$ of $\bF_{q^m}$ such that $H_r\cap \bF_{q^k}^ \perp=\{0\}$, which implies that $(H_r\cap \bF_{q^k}^ \perp)^\perp=\mathbb{F}_{q^m}$. Hence, $\theta_i \in (H_r\cap \bF_{q^k}^ \perp)^\perp$ for $0\leq i\leq h$. From (\ref{eq:050803}) and (\ref{eq:050804}), we have
\begin{equation}
\label{equ:Max1}
\mathop{\rm max}\, \left\{|D\cap H_r^{\perp}| \right\}=q^{m-r} - (h+1)q^{k-r}.
\end{equation}

\noindent {\bf Case 2:} $k<r \le m$. It is clear that $\theta_0=0\in(H_r\cap \bF_{q^k}^ \perp)^\perp$ and ${\rm dim}(H_r\cap\bF_{q^k}^\perp)\ge  r+(m-k)-m=r-k$. Then from (\ref{eq:050803}) and (\ref{eq:050804}), we obtain
\begin{equation}\label{eq:0511}
\mathop{\rm max}\, \left\{|D\cap H_r^{\perp}| \right\}=q^{m-r} - q^{k-r} \sum_{i=1}^hN_i-q^{k-r}N_0\le q^{m-r} - q^{k-r}|H_r\cap \bF_{q^k}^ \perp|\le q^{m-r}-1.
\end{equation}
From Lemma \ref{lm:unemptyset}, we know that there exists an $\alpha\in \bF_{q^m}$ such that $\alpha$ belongs to $\bF_{q^k}^ \perp$ and ${\rm Tr}_1^m(\theta_i\alpha)\ne 0$ for $1\le i\le h$.
It is obvious that there exists an $r$-dimensional subspace $H_r$ of $\bF_{q^m}$ such that $\alpha\in H_r$ and ${\rm dim}(H_r\cap\bF_{q^k}^\perp)=r-k$, then
$\alpha \in  \bF_{q^k}^ \perp\cap H_r$. From ${\rm Tr}_1^m(\theta_i\alpha)\ne 0$, we have $\theta_i\notin(H_r\cap \bF_{q^k}^ \perp)^\perp$ for $1\le i\le h$, which means that $N_i=0$ for any $1\le i\le h$.
Hence, the equality holds in (\ref{eq:0511}).

From the above two cases, the desired conclusion then follows.
\end{proof}

\begin{remark} In this section, we calculate the generalized Hamming weight of $\C_D$ in (\ref{eqn:linearcode1}) with the case $h\le q-1$. When $h$ is larger than $q$, the Hamming distance of $\C_D$ is influenced by the choice of $\theta_i$, and the generalized Hamming weight becomes more complex. Readers can refer to \cite{Hu2022} for results on the Hamming distance of $\C_D$ under specific situations when $h\ge q$. \end{remark}
\begin{example}
Let q=2, m=4, k=2, h=1. Magma experiments show that $\C_D$ is a [54,4,36] linear code, and the weight hierarchy of $\C_D$ is
$\{ wt_1=36,wt_2=48,wt_3=52,wt_4=54\}$, which is consistent with result in Theorem \ref{th:GHW1}.
\end{example}
\begin{example}
Let q=3, m=3, k=1, h=2. Magma experiments show that $\C_D$ is a [18,3,12] linear code, and the weight hierarchy of $\C_D$ is
$\{ wt_1=12,wt_2=16,wt_3=18\}$, which is consistent with result in Theorem \ref{th:GHW1}.
\end{example}

\subsection{Linear codes from the defining-set $(\bF_{q^m}\setminus \bF_{q^s})\times(\bF_{q^k}\setminus \bF_{q^\ell})$}

In this subsection, we determine the generalized Hamming weight  of linear code
\begin{equation}\label{eq:xy051104}
\C_D = \left\{ ({\rm Tr}_1^{m}(ax)+{\rm Tr}_1^{k}(by))_{(x,y) \in D}:a \in \bF_{q^m},b\in \bF_{q^k} \right\}
\end{equation}
with defining set
\begin{equation}\label{eq:xy051101}
D= \left\{ (x,y): x \in \bF_{q^m}\setminus \bF_{q^s}, y \in \bF_{q^k}\setminus \bF_{q^{\ell} } \right\},
\end{equation}
where $0<s<m$, $0<\ell<k$, $s|m$, $\ell \, |\, k$, $k-\ell\le m-s$ and $q^{m-s}>q^{m+\ell-k-s}+1$. It is clear that $\bF_{q^s}$ and $\bF_{q^{\ell}}$ are the subfields of $\bF_{q^m}$ and $\bF_{q^k}$, respectively. Let $H_r$ be a $r$-dimensional subspace of ${\mathbb F}_{q^m}\times{\mathbb F}_{q^k}$, where $1\leq r\leq m+k$. Recall that
\begin{equation}\label{eq:xy0511}
H^{\perp}_r=\left\{ (x,y)\in {\mathbb F}_{q^m}\times{\mathbb F}_{q^k}:{\rm Tr}_1^m(\alpha x)+{\rm Tr}_1^k(\beta y)=0\ {\rm for\ any}\ (\alpha,\beta)\in H_r \right\}.
\end{equation}
From Section 2,
we know that $H^{\perp}_r$ is an $(m+k-r)$-dimensional subspace of ${\mathbb F}_{q^m}\times{\mathbb F}_{q^k}$.
Then we have the following result.
\begin{lemma}
\label{prop:Commonzeros}
Let $D$ and $H_r^{\perp}$ be defined in (\ref{eq:xy051101}) and (\ref{eq:xy0511}). Then
$$|D\cap H_r^{\perp}|\le q^{m+k-r}-{\rm max}\{1, q^{m+\ell-r}\}.$$
\end{lemma}

\begin{proof}
 By the definition of $D$ and $H_r^{\perp}$, we have
$|D\cap H_r^{\perp}|\subset H_r^{\perp}\backslash\{(0,0)\},$
which implies that $|D\cap H_r^{\perp}|\le q^{m+k-r}-1$.

 When $1\le r \le m+\ell$, we have that
\begin{equation}\label{eq:051102}
\begin{split}
{\rm dim}(({\mathbb F}_{q^m}\times{\mathbb F}_{q^\ell})\cap H_r^{\perp})&\ge{\rm dim}({\mathbb F}_{q^m}\times{\mathbb F}_{q^\ell})+{\rm dim}(H_r^{\perp})-{\rm dim}({\mathbb F}_{q^m}\times{\mathbb F}_{q^k})\\
&= m+\ell+(m+k-r)-(m+k)\\
&=m+\ell-r.
\end{split}
\end{equation}
By the definition of $D$, it is clear that $D \cap({\mathbb F}_{q^m}\times{\mathbb F}_{q^{\ell}})=\{ (0,0) \}$. Then
$D\cap H_r^{\perp}\subset H_r^{\perp}\backslash(({\mathbb F}_{q^m}\times{\mathbb F}_{q^{\ell}})\cap H_r^{\perp})$.
Hence, from (\ref{eq:051102}) we have
$$
|D\cap H_r^{\perp}|\le |H_r^{\perp}|-|({\mathbb F}_{q^m}\times{\mathbb F}_{q^{\ell}})\cap H_r^{\perp}|\le q^{m+k-r}-q^{m+\ell-r}.
$$
This completes the proof.
\end{proof}

Next,  we determine the weight hierarchy of the second class of linear codes.

\begin{theorem} \label{th:GHW2}
Let $\C_D$ be a linear code defined in (\ref{eq:xy051104})
 with the defining set in (\ref{eq:xy051101}).
Then $\C_D$ is a $[(q^m-q^s)(q^k-q^{\ell}),m+k,(1-\frac{1}{q})(q^{m+k}-q^{m+\ell}-q^{k+s})]$ code  and has the following weight hierarchy: $$d_r(\C_D) =\begin{cases}
 (q^m-q^s)(q^k-q^{\ell})-q^{m+k-r}+q^{m+\ell-r}+q^{k+s-r}-q^{s+\ell},& \text{$1\le r\le k-\ell $},
\\
 (q^m-q^s)(q^k-q^\ell)-q^{m+k-r}+q^{m+\ell-r},& \text{$k-\ell<r \le m+\ell$},
\\
 (q^m-q^s)(q^k-q^\ell)-q^{m+k-r}+1,& \text{$m+\ell< r\le m+k$}.
\end{cases}$$
\end{theorem}

\begin{proof}
From \cite{Hu2021}, we know that the code $\C_D$ has parameters $[(q^m-q^s)(q^k-q^{\ell}),m+k,(1-\frac{1}{q})(q^{m+k}-q^{m+\ell}-q^{k+s})]$. Then
 the subspace
$$K={\left\{ (x,y) \in {\mathbb F}_{q^m}^2 \, :\, ( {\rm Tr}_1^m(d_1x+d_2y) )_{(d_1,d_2) \in D} = {\bf 0} \right\}}$$
is $\{(0,0)\}$.
From  Proposition \ref{GHW2}, it follows that
$$
d_r(\C_D) =n-{\rm max}\, \left\{ |D\cap H^\perp| : H \in [{\mathbb F}_{q^m}\times\bF_{q^k}, r] \right\}.
$$
We now determine the maximum value of $|D\cap H^{\perp}|$ for $H \in [{\mathbb F}_{q^m}\times\bF_{q^k}, r]$.
Let $H_r$ be an $r$-dimensional subspace of ${\mathbb F}_{q^m}\times{\mathbb F}_{q^k}$.  Then
\begin{equation}\label{eqn:commonzerosbit}
\begin{split}
|D\cap H_r^{\perp}| =& \frac{1}{q^r}\sum_{\left ( \alpha ,\beta  \right ) \in H_r}\sum_{x\in\bF_{q^m}\backslash\bF_{q^s}}\sum_{y\in\bF_{q^k}\backslash\bF_{q^{\ell}}}\zeta_q^{({\rm Tr}_1^m(\alpha x)+{\rm Tr}_1^k(\beta y))}\\
                     =&\frac{1}{q^r} \sum_{\left ( \alpha ,\beta  \right ) \in H_r}N(\alpha,\beta),
\end{split}
\end{equation}
where
$$N(\alpha,\beta)=\left ( \sum_{x\in \bF_{q^m} }\zeta
_q^{{\rm Tr}^m_1\left ( \alpha x \right ) } -  \sum_{x\in \bF_{q^s} }\zeta
_q^{{\rm Tr}^m_1\left ( \alpha x \right ) }  \right )
\left ( \sum_{x\in \bF_{q^k} }\zeta
_q^{{\rm Tr}^k_1\left ( \beta x \right ) } -  \sum_{x\in \bF_{q^{\ell}} }\zeta
_q^{{\rm Tr}^k_1\left ( \beta x \right ) } \right ).$$
Let $\mathcal{D}_1=\{(\alpha,0)\in \bF_{q^m}\times\bF_{q^k}:\alpha\in\bF_{q^s}^\perp\backslash\{0\}\}$, $\mathcal{D}_2=\{(0,\beta)\in \bF_{q^m}\times\bF_{q^k}: \beta\in\bF_{q^\ell}^\perp\backslash\{0\}\}$ and $\mathcal{D}_3=\{(\alpha,\beta)\in \bF_{q^m}\times\bF_{q^k}:\alpha\in\bF_{q^s}^\perp\backslash\{0\}, \beta\in\bF_{q^\ell}^\perp\backslash\{0\}\}$, where
$$
\bF_{q^s}^\perp=\{x\in \bF_{q^m}: {\rm Tr}_s^m(x)=0\}\,\,\text{and}\,\,\bF_{q^\ell}^\perp=\{x\in \bF_{q^k}: {\rm Tr}_\ell^k(x)=0\}.
$$
 It is easy to check that
\begin{eqnarray}
N(\alpha,\beta)
&=&
\begin{cases}
 (q^m - q^s)(q^k- q^{\ell}),& \text{$\alpha = 0, \beta = 0 $},
\\
-(q^k- q^\ell)q^s,& \text{$(\alpha,\beta) \in H_r\cap \mathcal{D}_1$},
\\
-(q^m - q^s)q^{\ell},& \text{$(\alpha,\beta) \in H_r\cap \mathcal{D}_2$},
\\
q^{s+\ell},&\text{$(\alpha,\beta) \in H_r\cap \mathcal{D}_3$},
\\
0,& \text{others.}
\end{cases}\label{eqn:Nalphabeta}
\end{eqnarray}
 According to
(\ref{eqn:commonzerosbit}) and (\ref{eqn:Nalphabeta}), we have
\begin{equation}\label{eqn:commonzeros2}
\begin{split}
&|D\cap H_r^{\perp}|\\
&=\frac{1}{q^r}\bigg((q^m-q^s)(q^k-q^\ell)-(q^k-q^\ell)q^s | H_r\cap\mathcal{D}_1|-(q^m-q^s)q^\ell|H_r\cap\mathcal{D}_2|+q^{s+\ell}|H_r\cap\mathcal{D}_3|\bigg).
\end{split}
\end{equation}
Let
\begin{equation}\label{eq:051501}
\{u_1,u_2,\dots,u_m\}\,\,\text{and}\,\,\{v_1,v_2,\dots,v_k\}
\end{equation}  be an ${\mathbb F}_q$-basis of ${\mathbb F}_{q^m}$ and
an ${\mathbb F}_q$-basis of ${\mathbb F}_{q^k}$ such that
\begin{equation}\label{eq:051502}
\{u_1,\dots,u_{m-s}\} \,\,\text{and}\,\, \{v_1,\dots,v_{k-\ell}\}
\end{equation}
form an ${\mathbb F}_q$-basis of $\bF_{q^s}^\perp$ and  an ${\mathbb F}_q$-basis of $\bF_{q^\ell}^\perp$, respectively. There are three cases for discussion.

\noindent {\bf Case 1:} $1\le r\le k-\ell$. Since $|H_r|=q^r$ and $(0,0) \notin \mathcal{D}_3$, we know that $|H_r\cap\mathcal{D}_3|\leq q^r-1$. Then from (\ref{eqn:commonzeros2}), we have
\begin{equation}\label{eq:051503}
|D\cap H_r^{\perp}|\le \frac{1}{q^r}\left((q^m-q^s)(q^k-q^\ell)+q^{s+\ell}(q^r-1)\right).
\end{equation}
Since $m-s\geq k-\ell$, we have $r \leq {\rm min}\{k-\ell,m-s\}$. It is clear that $\{(u_1,v_1),(u_2,v_2),\dots,(u_r,v_r)\}$ is $\mathbb{F}_q$-linear independent in ${\mathbb F}_{q^m}\times\bF_{q^k}$, where $u_i$ and $v_i$ are in (\ref{eq:051502}) for $1\leq i\leq r$. Let
\begin{equation}\label{eq:hr}
H_r={\rm span}\{(u_1,v_1),(u_2,v_2),\dots,(u_r,v_r)\}.
 \end{equation}
One has $|H_r\cap\mathcal{D}_1|=|H_r\cap\mathcal{D}_2|=0$ and $|H_r\cap\mathcal{D}_3|=q^r-1$ since $H_r\backslash \{(0,0)\}\subset \mathcal{D}_3$. Then  from (\ref{eqn:commonzeros2}) we have
\begin{equation*}
|D\cap H_r^{\perp}|= \frac{1}{q^r}\left((q^m-q^s)(q^k-q^\ell)+q^{s+\ell}(q^r-1)\right).
\end{equation*}
So, $H_r$ given in (\ref{eq:hr}) is such that the equality in (\ref{eq:051503}) holds. Hence, the $r$-th generalized Hamming weight of $\C_D$ is
$d_r(\C_D)= (q^m-q^s)(q^k-q^\ell)-q^{m+k-r}+q^{m+\ell-r}+q^{k+s-r}-q^{s+\ell}$ for $1\le r\le k-\ell$.

\noindent {\bf Case 2:} $k-\ell<r \le m+\ell$.
According to Proposition \ref{prop:Commonzeros} we have
\begin{equation}\label{eq:0516}
\mathop{\rm max}\, \left\{|D\cap H_r^{\perp}| \right\}\le q^{m+k-r}-q^{m+\ell-r}.
\end{equation}
There are three subcases for discussion.
\begin{enumerate}
\item[(1)] $k-\ell<r \le k$. Let $T_1=\{(u_1,v_1),(u_2,v_2),\dots,(u_{k-\ell},v_{k-\ell})\}$, where $u_i$ and $v_i$ are in (\ref{eq:051502}) for $1\leq i\leq k-\ell$. It is easy to see that $\{T_1,(0,v_{k-\ell+1}),\dots,(0,v_k)\}$ is
$\mathbb{F}_q$-linear independent in ${\mathbb F}_{q^m}\times\bF_{q^k}$, where $v_i$ is in (\ref{eq:051501}) for $k-\ell +1\leq i\leq k$. Let \begin{equation}\label{eq:hr01}
H_r={\rm span}\{T_1,(0,v_{k-\ell+1}),\dots,(0,v_k)\}.
\end{equation}
By the definition of $\mathcal{D}_1$, $\mathcal{D}_2$ and $\mathcal{D}_3$, we have $|H_r\cap\mathcal{D}_1|=|H_r\cap\mathcal{D}_2|=0$, $|H_r\cap\mathcal{D}_3|=q^{k-\ell}-1$. From (\ref{eqn:commonzeros2}) we obtain
\begin{equation*}
|D\cap H_r^{\perp}|=\frac{1}{q^r}\left((q^m-q^s)(q^k-q^\ell)+q^{s+\ell}(q^{k-\ell}-1)\right)= q^{m+k-r}-q^{m+\ell-r}.
\end{equation*}
So, $H_r$ given in (\ref{eq:hr01}) in this case is such that the equality in (\ref{eq:0516}) holds.
\item[(2)] $k<r\le k+s$. Let $T_2=\{(0,v_{k-\ell+1}),(0,v_{k-\ell+2}),\dots,(0,v_{k})\}$ and
\begin{equation}\label{eq:hr02}
H_r={\rm span}\{T_1,T_2,(u_{m-s+1},0),\dots,(u_{r-s+m-k},0)\}.
\end{equation}
With an analysis similar as above, we have
\begin{eqnarray*}
|D\cap H_r^{\perp}|&=& q^{m+k-r}-q^{m+\ell-r}.
\end{eqnarray*}
So, $H_r$ given in (\ref{eq:hr02}) in this case is such that the equality in (\ref{eq:0516}) holds.

\item[(3)] $k+s<r\le m+\ell$. Let
$$T_3=\{(u_{k-\ell+1},0),(u_{k-\ell+2},0),\dots,(u_{r-\ell-s},0)\}\,\, \text{and}\,\, T_4=\{(u_{m-s+1},0),(u_{m-s+2},0),\dots,(u_{m},0)\},$$
where $u_i$ is in (\ref{eq:051502}) for $k-\ell+1\leq i\leq r-\ell-s$ and $m-s+1\leq i\leq m$.  It is clear that $\{T_1,T_2,T_3,T_4\}$ is
$\mathbb{F}_q$-linear independent in  ${\mathbb F}_{q^m}\times\bF_{q^k}$. Let
\begin{equation}\label{eq:hr03}
H_r={\rm span}\{T_1,T_2,T_3,T_4\}.
\end{equation}
By the definition of $\mathcal{D}_1$, $\mathcal{D}_2$ and $\mathcal{D}_3$, we have $|H_r\cap\mathcal{D}_1|=q^{r-k-s}-1$, $|H_r\cap\mathcal{D}_2|=0$ and $|H_r\cap\mathcal{D}_3|=q^{r-\ell-s}-q^{r-k-s}$. Then
\begin{equation*}
\begin{split}
|D\cap H_r^{\perp}|=&\frac{1}{q^r}\left((q^m-q^s)(q^k-q^\ell)+(-(q^k-q^\ell)q^s)\cdot(q^{r-k-s}-1)+q^{s+\ell}\cdot(q^{r-\ell-s}-q^{r-k-s})\right)\\
=& q^{m+k-r}-q^{m+\ell-r}.
\end{split}
\end{equation*}
So, $H_r$ given in (\ref{eq:hr03}) in this case is such that the equality in (\ref{eq:0516}) holds.
\end{enumerate}

From above results, we know that there exists an $r$-dimensional subspace $H_r$ of ${\mathbb F}_{q^m}\times\bF_{q^k}$ such that the equality in (\ref{eq:0516}) holds, which implies that
the $r$-th generalized Hamming weight of $\C_D$ is $d_r(\C_D)= (q^m-q^s)(q^k-q^\ell)-q^{m+k-r}+q^{m+\ell-r}$ for $k-\ell<r \le m+\ell$.

\noindent {\bf Case 3:} $m+\ell< r\le m+k$. According to Proposition \ref{prop:Commonzeros} we have
\begin{equation}\label{eq:0525}
\mathop{\rm max}\, \left\{|D\cap H_r^{\perp}| \right\}\le q^{m+k-r}-1.
\end{equation}
Denote $T_5=\left\{ (u_{1},0),\dots,(u_{r-m-\ell},0)\right\}$ and $T_6=\{(u_{k-\ell+1},0),(u_{k-\ell+2},0),\dots,(u_{m-s},0)\}$. Let
\begin{equation}\label{eq:hr04}
H_r={\rm span}\{T_1,T_2,T_3,T_5,T_6\}.
\end{equation}
By the definition of $\mathcal{D}_1$, $\mathcal{D}_2$ and $\mathcal{D}_3$, we have $|H_r\cap\mathcal{D}_1|=q^{r-k-s}-1$, $|H_r\cap\mathcal{D}_2|=q^{r-m-\ell}-1$, $|H_r\cap\mathcal{D}_3|=q^{r-\ell-s}-q^{r-m-\ell}-q^{r-k-s}+1$. Then
\begin{eqnarray*}
|D\cap H_r^{\perp}|&=&\frac{1}{q^r}\Big((q^m-q^s)(q^k-q^\ell)+(-(q^k-q^\ell)q^s)\cdot(q^{r-k-s}-1)\\
                     & & +(-(q^m-q^s)q^\ell)\cdot(q^{r-m-\ell}-1)+q^{s+\ell}(q^{r-\ell-s}-q^{r-m-\ell}-q^{r-k-s}+1)\Big)\\
                     &=& q^{m+k-r}-1.
\end{eqnarray*}
This shows that there exists an $r$-dimensional subspace $H_r$ of ${\mathbb F}_{q^m}\times\bF_{q^k}$ such that the equality in (\ref{eq:0525}) holds. Then the $r$-th generalized Hamming weight of $\C_D$ is
$d_r(\C_D)= (q^m-q^s)(q^k-q^\ell)-q^{m+k-r}+1$.

Combining the above three cases, the desired conclusion then follows.
\end{proof}

\begin{example}
Let $q=2, m=3, k=2, s=1, \ell=1$. Magma experiments show that $\C_D$ is a [12,5,4] linear code, and the weight hierarchy of $\C_D$ is
$\{wt_1=4,wt_2=8,wt_3=10,wt_4=11,wt_5=12\}$, which is consistent with result in Theorem \ref{th:GHW2}.
\end{example}
\begin{example}
Let $q=3, m=2, k=2, s=1, \ell=1$. Magma experiments show that $\C_D$ is a [36,4,18] linear code, and the weight hierarchy of $\C_D$ is
$\{ wt_1=18,wt_2=30,wt_3=34,wt_4=36\}$, which is consistent with result in Theorem \ref{th:GHW2}.
\end{example}

\begin{remark} To ensure the linear codes given in (\ref{eq:xy051104}) is nontrivial, the defining set in (\ref{eq:xy051101}) needs to satisfy that $q^{m-s}>q^{m+\ell-k-s}+1$.
In fact, this condition is always true except when $q=m=k=2$ and $s=\ell=1$. In this exceptional case, the linear  code
$$ \C= \left\{ ({\rm Tr}_1^{2}(ax)+{\rm Tr}_1^{2}(by))_{(x,y) \in (\bF_{2^2}\setminus \bF_{2})\times(\bF_{2^2}\setminus \bF_{2})}:a \in \bF_{2^2},b\in \bF_{2^2} \right\}$$
is a $[4,3,2]$ code with the weight hierarchy $\{wt_1=2,wt_2=3,wt_3=4\}$.
\end{remark}

\subsection{Linear codes from a defining-set of butterfly structure}

In this subsection, we will investigate the weight hierarchy of the linear code
\begin{equation}
\label{eqn:linearcode3}
\C_D = \{({\rm Tr}_1^m(ax+by))_{(x,y) \in D}:a,b \in \bF_{2^m} \},
\end{equation}
where the defining set is
\begin{equation}\label{eq:0526}
D = \{(x,y) \in \bF_{2^m}^2 : ({\rm Tr}_1^m (x(y+1)), {\rm Tr}_1^m(y(x+1))) = (0,1)\}.
\end{equation}
Here, $D$ is called butterfly structure in \cite{Canteaut2017,Perrin2016,Hu2021}.
In order to obtain the  weight hierarchy of $\C_D$, we need the following proposition.

\begin{proposition}
\label{prop:ghw3}
Let $\C_D$ be a linear code defined in (\ref{eqn:linearcode3}), then the $r$-th  generalized Hamming weight of $\C_D$ is
\begin{eqnarray*}
d_r(\C_D)=&2^{2m-2}-\mathop{\rm max}\Bigg\{\frac{1}{2^{r+2}}\Bigg(2^{2m}+2^m\sum_{(\alpha,\beta)\in H_r}((-1)^{{\rm Tr}_1^m(\beta(\alpha+1))} - (-1)^{{\rm Tr}_1^m(\alpha(\beta+1))})\\
          &-2^{2m}{\rm I}_{(1,1)}(H_r)\Bigg)\Bigg\},
\end{eqnarray*}
where $1\leq r\leq 2m$, $H_r$ is an $r$-dimensional subspace of $\mathbb{F}_{2^m}^2$ and ${\rm I}_{(1,1)}(H_r)$ is a discriminant function, i.e.
$$
{\rm I}_{(1,1)}(H_r)=
\begin{cases}
 1,& (1,1)\in H_r,
\\
0,& (1,1)\notin H_r.
\end{cases}
$$
\end{proposition}
\begin{proof}
From \cite{Hu2021}, we know that the code $\C_D$ has parameters $[2^{2m-2},2m,2^{2m-3}-2^{m-2}]$. Then
 the subspace
$$K={\left\{ (x,y) \in {\mathbb F}_{q^m}^2 \, :\, ( {\rm Tr}_1^m(d_1x+d_2y) )_{(d_1,d_2) \in D} = {\bf 0} \right\}}$$
is $\{(0,0)\}$. Thus by  Proposition \ref{GHW2}, the $r$-th generalized Hamming weight of $\C_{D}$ is
\begin{equation}\label{eq:rghw2}
d_r(\C_D) =n-{\rm max}\, \left\{ |D\cap H^\perp| : H \in [{\mathbb F}_{2^m}^2, r] \right\},
\end{equation}
where
$$
H^{\perp}=\left\{ (x,y)\in {\mathbb F}_{2^m}^2\,:\,{\rm Tr}_1^m(\alpha x+\beta y)=0\ {\rm for\ any}\ (\alpha,\beta)\in H \right\}.
$$
Let $H_r$ be an $r$-dimensional subspace of $\mathbb{F}_{2^m}^2$. By the definition of $D$ and the dual of an $r$-dimensional subspace, we have that
\begin{equation}\label{eq:052503}
\begin{split}
|D\cap H_r^\perp|=&\frac{1}{2^{r+2}}\sum_{x,y\in \bF_{2^m}}\sum_{(\alpha,\beta)\in H_r}\Bigg(\left(1+(-1)^{{\rm Tr}_1^m (x(y+1))}\right)\left(1-(-1)^{{\rm Tr}_1^m (y(x+1))}\right)(-1)^{{\rm Tr}_1^m(\alpha x+\beta y)}\Bigg)\\
                   =&\frac{1}{2^{r+2}}\sum_{(\alpha,\beta)\in H_r}\sum_{x,y\in \bF_{2^m}}\Bigg((-1)^{{\rm Tr}_1^m(\alpha x+\beta y)}+(-1)^{{\rm Tr}_1^m (x(y+1)+\alpha x+\beta y)}\\
                   & -(-1)^{{\rm Tr}_1^m (y(x+1)+\alpha x+\beta y)}-(-1)^{{\rm Tr}_1^m(x(y+1)+y(x+1)+\alpha x+\beta y)}\Bigg).
\end{split}
\end{equation}
Obviously,
\begin{equation}\label{eq:052504}
\begin{split}
\sum_{(\alpha,\beta)\in H_r}\sum_{x,y\in \bF_{2^m}}(-1)^{{\rm Tr}_1^m (x(y+1)+\alpha x+\beta y)}&=\sum_{(\alpha,\beta)\in H_r}\sum_{x\in \bF_{2^m}}(-1)^{{\rm Tr}_1^m (x(\alpha+1))}\sum_{y\in \bF_{2^m}}(-1)^{{\rm Tr}_1^m (y(\beta+x))}\\
&=2^m\sum_{(\alpha,\beta)\in H_r}(-1)^{{\rm Tr}_1^m(\beta(\alpha+1))}.
\end{split}
\end{equation}
Similarly, we have
\begin{equation}\label{eq:052505}
\begin{split}
\sum_{(\alpha,\beta)\in H_r}\sum_{x,y\in \bF_{2^m}}(-1)^{{\rm Tr}_1^m (y(x+1)+\alpha x+\beta y)}=2^m\sum_{(\alpha,\beta)\in H_r}(-1)^{{\rm Tr}_1^m(\alpha(\beta+1))}.
\end{split}
\end{equation}
In addition, we know that
\begin{equation}\label{eq:052506}
\begin{split}
\sum_{(\alpha,\beta)\in H_r}\sum_{x,y\in \bF_{2^m}}(-1)^{{\rm Tr}_1^m(x(y+1)+y(x+1)+\alpha x+\beta y)}&=\sum_{(\alpha,\beta)\in H_r}\sum_{x\in \bF_{2^m}}(-1)^{{\rm Tr}_1^m (x(\alpha+1))}\sum_{y\in \bF_{2^m}}(-1)^{{\rm Tr}_1^m (y(\beta+1))}\\
&=
\begin{cases}
2^{2m},& (1,1)\in H_r,
\\
0,& (1,1)\notin H_r.
\end{cases}
\end{split}
\end{equation}
From (\ref{eq:rghw2})-(\ref{eq:052506}), the desired conclusion then follows.
\end{proof}

From Proposition \ref{prop:ghw3}, to obtain the $r$-th generalized Hamming weight of $\C_D$, we are required to find an $r$-dimensional subspace $H$ of $\bF_{2^m}^2$ that does not contain $(1,1)$, but contains as many as possible
points $(\alpha, \beta)\in H$ satisfying $({\rm Tr}_1^m(\beta(\alpha+1)),{\rm Tr}_1^m(\alpha(\beta+1))) = (0,1)$ . To archive this, we need two lemmas below.

\begin{lemma}\label{lmm:differeceofcharacter1}
Let $H_r$ be an $r$-dimensional subspace of $\mathbb{F}_{2^m}^2$. Then
$$
|\{(\alpha, \beta)\in H_r: {\rm Tr}_1^m(\beta(\alpha+1))=0, {\rm Tr}_1^m(\alpha(\beta+1))=1\}|\le 2^{r-1}.
$$

\end{lemma}
\begin{proof}
Let $\varphi $ be a map from $\bF_{2^m}^2$ to $\bF_{2}^2$, which is given as follows
\begin{equation*}
\begin{split}
\varphi:\,\,\,  &\bF_{2^m}\times\bF_{2^m}  \longrightarrow   \bF_{2}\times\bF_{2} \\
   &     (x,y) \,\,\,\,\,\,\,\,\,\,\,\,\,\,\,\,\longmapsto ({\rm Tr}_1^m(y(x+1)),{\rm Tr}_1^m(x(y+1)))\mbox{.}
\end{split}
\end{equation*}
Let $V$ be the set
$$
V=\{(\alpha, \beta)\in H_r: \varphi(\alpha,\beta)=(0,1)\}.
$$
If there exist $(\alpha_1,\beta_1),(\alpha_2,\beta_2)\in V$ , then we have
$${\rm Tr}_1^m(\alpha_1\beta_1+\alpha_1)={\rm Tr}_1^m(\alpha_2\beta_2+\alpha_2)=1$$
and
$${\rm Tr}_1^m(\alpha_1\beta_1+\beta_1)={\rm Tr}_1^m(\alpha_2\beta_2+\beta_2)=0.$$
These imply that
$$\varphi(\alpha_1+\alpha_2,\beta_1+\beta_2)=({\rm Tr}_1^m(\alpha_1\beta_2+\alpha_2\beta_1),{\rm Tr}_1^m(\alpha_1\beta_2+\alpha_2\beta_1))\neq(0,1).$$
Hence, we have $(\alpha_1+\alpha_2,\beta_1+\beta_2)\notin V$.  This means that  if $V$ is not a empty set, then $(v+V)\cap V=\emptyset$ for any $v\in V$.
Since $V$ and $v+V$ are subsets of $H_r$, we can get $2|V|\le|H_r|=2^r$. This completes the proof.
\end{proof}

\begin{lemma}
\label{lmm:differeceofcharacter2}
Let the notation be given as in Proposition \ref{prop:ghw3}, then
$$-2^m\le\sum_{(\alpha, \beta) \in H_r}\left((-1)^{{\rm Tr}_1^m(\beta(\alpha+1))} - (-1)^{{\rm Tr}_1^m(\alpha(\beta+1))}\right) \le 2^m.$$ Specifically,
$$\sum_{(\alpha, \beta) \in H_r}\left((-1)^{{\rm Tr}_1^m(\beta(\alpha+1))} - (-1)^{{\rm Tr}_1^m(\alpha(\beta+1))}\right)=0$$
if $(1,1)\in H_r$.
\end{lemma}

\begin{proof}
 According to Proposition \ref{prop:ghw3}, we have
\begin{equation}\label{eq:0520}
|D\cap H_r^\perp|=\frac{1}{2^{r+2}}\left(2^{2m}+2^m\sum_{(\alpha,\beta)\in H_r}\left((-1)^{{\rm Tr}_1^m(\beta(\alpha+1))} - (-1)^{{\rm Tr}_1^m(\alpha(\beta+1))}\right)-2^{2m}{\rm I}_{(1,1)}(H_r)\right).
\end{equation}
There are two cases for discussion.

\noindent {\bf Case 1:} $(1,1)\notin H_r$. It is clear that $I_{(1,1)}(H_r)=0$. Since $|D\cap H_r^\perp|\geq 0$, from (\ref{eq:0520}) we have
\begin{equation}\label{eq:052001}
\sum_{(\alpha, \beta) \in H_r}((-1)^{{\rm Tr}_1^m(\beta(\alpha+1))} - (-1)^{{\rm Tr}_1^m(\alpha(\beta+1))}) \ge -2^m.
\end{equation}
Let
\begin{equation}\label{eq:052002}
H'_r=\{(\alpha,\beta)\in\bF_{2^m}^2:\,(\beta,\alpha)\in H_r\}.
\end{equation}
Then $H'_r$ is also an $\bF_2$-subspace of $\bF_{2^m}^2$ and $(1,1)\notin H'_r$. Thus it holds that
\begin{equation}\label{eq:052003}
\sum_{(\alpha, \beta) \in H'_r}((-1)^{{\rm Tr}_1^m(\beta(\alpha+1))} - (-1)^{{\rm Tr}_1^m(\alpha(\beta+1))}) \ge -2^m.
\end{equation}
From (\ref{eq:052001})-(\ref{eq:052003}), we have
$$-2^m\leq \sum_{(\alpha, \beta) \in H_r}((-1)^{{\rm Tr}_1^m(\beta(\alpha+1))} - (-1)^{{\rm Tr}_1^m(\alpha(\beta+1))})\leq 2^m.$$

\noindent {\bf Case 2:} $(1,1)\in H_r$. With an analysis similar as above,
 we can obtain that
$$
0\le \sum_{(\alpha, \beta) \in H_r}((-1)^{{\rm Tr}_1^m(\beta(\alpha+1))} - (-1)^{{\rm Tr}_1^m(\alpha(\beta+1))})\le 0,
$$
i.e. $\sum_{(\alpha, \beta) \in H_r}((-1)^{{\rm Tr}_1^m(\beta(\alpha+1))} - (-1)^{{\rm Tr}_1^m(\alpha(\beta+1))})=0$. This completes the proof.
\end{proof}

In the following,  we determine the weight hierarchy of the third class of linear codes.

\begin{theorem}
\label{thm:GHW3}
Let $\C_D$ be a linear code defined in (\ref{eqn:linearcode3}) with the defining set in (\ref{eq:0526}).
 Then $\C_D$ is a $[2^{2m-2},2m,2^{2m-3}-2^{m-2}]$ code and the $r$-th generalized Hamming weight of $\C_D$ is given by $$d_r(\C_D) =\begin{cases}
 2^{2m-2}-2^{2m-r-2}-2^{m-2},& \text{$1\le r\le m $},
\\
2^{2m-2}-2^{2m-r-1},& \text{$m<r < 2m$},
\\
2^{2m-2},& \text{$r = 2m$}.
\end{cases}$$

\end{theorem}
\begin{proof}
From \cite{Hu2021}, we know that the code $\C_D$ has parameters $[2^{2m-2},2m,2^{2m-3}-2^{m-2}]$. From Proposition \ref{prop:ghw3}, the $r$-th generalized Hamming weight of $\C_D$ can be derived that
\begin{equation}\label{eq:052005}
d_r(\C_D) =n-{\rm max}\, \left\{ |D\cap H^\perp| : H \in [{\mathbb F}_{2^m}^2, r] \right\},
\end{equation}
where
\begin{equation}
\label{equ:commonzeros3}
|D\cap H^\perp|=\frac{1}{2^{r+2}}\left(2^{2m}+2^m\sum_{(\alpha,\beta)\in H}((-1)^{{\rm Tr}_1^m(\beta(\alpha+1))} - (-1)^{{\rm Tr}_1^m(\alpha(\beta+1))})-2^{2m}{\rm I}_{(1,1)}(H)\right).
\end{equation}
Let $H_r$ be an $r$-dimensional $\bF_2$-subspace of $\bF_{2^m}^2$. If $r=2m$, i.e., $H_r=\mathbb{F}_{2^m}^2$, it is clear that $(1,1)\in H_r$. According to  Lemma \ref{lmm:differeceofcharacter2} and (\ref{equ:commonzeros3}), we have $|D\cap H_r^\perp|=0$. From (\ref{eq:052005}) we have $d_{2m}(\C_D)=2^{2m-2}$. If $1\le r<2m$, from Lemmas \ref{lmm:differeceofcharacter1} and \ref{lmm:differeceofcharacter2} we obtain
\begin{equation}
\label{eqn:commonzeros4}
|D\cap H^\perp_r|\le \frac{1}{2^{r+2}}(2^{2m}+2^m\cdot {\rm min}\{2^r,2^m\})=2^{2m-r-2}+2^{m-r-2+{\rm min}\{r,m\}}.
\end{equation}
In the following, we aim to construct the subspace $H_r$ such that the equality holds in (\ref{eqn:commonzeros4}). There are two cases for discussion.

\noindent {\bf Case 1:} $1\le r \le m$. Let $V$ be an $r$-dimensional $\bF_2$-subspace of $\bF_{2^m}$ such that there exists a $v\in V$ with ${\rm Tr}_1^m(v)=1$. It follows that $|\{x\in V:\,{\rm Tr}_1^m(x)=1\}|=2^{r-1}$ since $\{x\in V:\,{\rm Tr}_1^m(x)=1\}+v= \{x\in V:\,{\rm Tr}_1^m(x)=0\}$. Denote $H$ as the set $H=\{(x,0):\,x\in V\}$ and $H$ will be an $r$-dimensional $\bF_2$-subspace of $\bF_{2^m}^2$. Then from (\ref{equ:commonzeros3}) we know that there exists an $r$-dimensional $H_r$ such that
$$
|D\cap H_r^\perp|=\frac{1}{2^{r+2}}\left(2^{2m}+2^m\sum_{u\in V}(1-(-1)^{{\rm Tr}_1^m(u)})\right)=2^{2m-r-2}+2^{m-2}.
$$
Consequently, we have $$\mathop{\rm max}\limits_{H\in [\bF_{2^m}^2,r]}\, \left\{|D\cap H^\perp| \right\}= 2^{2m-r-2}+2^{m-2},$$ and the $r$-th generalized Hamming weight of $\C_D$ is $$d_r(\C_D)=2^{2m-2}-2^{2m-r-2}-2^{m-2}.$$

\noindent {\bf Case 2:} $m<r<2m$. Let $V$ be an $(r-m)$-dimensional $\bF_2$-subspace of $\bF_{2^m}$ such that $1\notin V$. Denote $H$ as the set $H=\{(x,y):\,x\in \bF_{2^m}, y\in V\}$. Then, $H$ is an $r$-dimensional $\bF_2$-subspace of $\bF_{2^m}^2$. We have
\begin{eqnarray*}
|D\cap H^\perp|&=&\frac{1}{2^{r+2}}\left(2^{2m}+2^m\sum_{\alpha\in\bF_{2^m},\beta\in V}((-1)^{{\rm Tr}_1^m(\beta(\alpha+1))} - (-1)^{{\rm Tr}_1^m(\alpha(\beta+1))})\right)\\
                   &=&\frac{1}{2^{r+2}}\left(2^{2m}+2^m\sum_{\alpha\in\bF_{2^m},\beta=0}(-1)^{{\rm Tr}_1^m(\beta(\alpha+1))}\right)\\
                   &=&2^{2m-r-1}.
\end{eqnarray*}
Thus, from (\ref{eqn:commonzeros4}) we can conclude that there exists an $r$-dimensional $H_r$ such that  $$ |D\cap H_r^\perp|= 2^{2m-r-1}$$
and the $r$-th generalized Hamming weight of $\C_D$ is $$d_r(\C_D)=2^{2m-2}-2^{2m-r-1}.$$
This completes the proof.
\end{proof}
\begin{example}
Let $m=2$. Magma experiments show that $\C_D$ is a [4,4,1] linear code, and the weight hierarchy of $\C_D$ is
$\{ wt_1=1,wt_2=2,wt_3=3,wt_4=4\}$, which is consistent with result in Theorem \ref{thm:GHW3}.
\end{example}
\begin{example}
Let $m=3$. Magma experiments show that $\C_D$ is a [16,6,6] linear code, and the weight hierarchy of $\C_D$ is
$\{ wt_1=6,wt_2=10,wt_3=12,wt_4=14,wt_5=15,wt_6=16\}$, which is consistent with result in Theorem \ref{thm:GHW3}.
\end{example}
\begin{remark}
By selecting the defining set $D=\{(x,y) \in \bF_{2^m}^2 : ({\rm Tr}_1^m (x(y+1)), {\rm Tr}_1^m(y(x+1))) = (1,0)\}$, the same result as Theorem \ref{thm:GHW3} will be obtained. However, when choosing the defining set $D=\{(x,y) \in \bF_{2^m}^2 : ({\rm Tr}_1^m (x(y+1)), {\rm Tr}_1^m(y(x+1))) = (0,0)\}$ or altering $(0,0)$ to $(1,1)$, the dimension of the code $\C_D$ is deviates from $2m$, posing challenges in calculating the GHWs of $\C_D$.
\end{remark}

\section{Concluding remarks}
In this paper, we determined the weight hierarchies of three classes of linear codes from defining set construction. The main ideas of this paper are as follows. We first obtain the low bound on the generalized Hamming weight of the considered codes and then construct the subspaces of the discussed codes which archive the bound. Interested readers are cordially invited to investigate the weight hierarchies of linear codes constructed from different defining-sets through this approach.

\end{document}